\documentclass[11pt]{article}

\usepackage{amsmath}
\usepackage{amsthm}
\usepackage{amsfonts}
\usepackage[a4paper,margin=25mm,bottom=35mm]{geometry}
\usepackage{graphicx}
\usepackage{paralist}
\usepackage{hyperref}
\usepackage{cleveref}

\newtheorem{lemma}{Lemma}
\newtheorem{theorem}[lemma]{Theorem}
\theoremstyle{definition}

\newtheorem{definition}{Definition}
\DeclareMathOperator{\area}{area}
\newcommand{\anc}{anchored}
\newcommand{\qu}{quadrilateral}
\newcommand{\para}{parallelogram}
\newcommand{\uu}{\mathbf{u}}

\newcommand{\ldt}{\mathrel{.\,.}}

\begin{document}

\title{%A linear-time algorithm for 
  The Largest
 Contained
  Quadrilateral
 and \\ the Smallest Enclosing Parallelogram
of 
a Convex Polygon}
\author{G\"unter Rote}
\date{\today}
\maketitle
 \begin{abstract}
   We present a linear-time algorithm for finding the
   quadrilateral of largest area contained in a convex
   polygon, and we show that it is closely related to an old
   algorithm for the smallest enclosing parallelogram of a convex
   polygon.
\end{abstract}

\section{Introduction}
A linear-time algorithm for the \emph{largest quadrilateral contained in a
convex polygon} was proposed in 1979 by Dobkin and
Snyder~\cite{dobkin79}.  This algorithm stood until 2017, when
Keikha, L\"offler, Mohades, Urhausen, and van der Hoog~\cite{KLUvdH17}
constructed a counterexample for which
it % the algorithm of Dobkin and Snyder~\cite{dobkin79} 
fails.

A simple linear-time algorithm for the
\emph{smallest parallelogram
enclosing a convex polygon}
was published
in a technical report by
Schwarz, Teich, Welzl, and Evans~\cite{stwe-fmep-94} in 1994, see also
\cite{stvwe-mepa-95}.

We will show that the two problems are closely related, in particular
when they are constrained by \emph{anchoring} them to some specified direction.
The solution of one problem provides an optimality certificate for the
other problem.
We present a conceptually simple algorithm that
treats both problems
 in a symmetric way and solves them
simultaneously in linear time.
%The algorithm is both
%conceptually simple and simple to implement.
%enough
The algorithm is based on the ``rotating calipers'' technique
from the early days of computational geometry. Proofs are included,
so that there can be no doubts about its correctness.

The algorithm becomes very simple when specialized for
solving only one of the two problems,
 see \Cref{sec:4gon,%}.
%\Cref{
sec:schwarz}.
Linear-time algorithms for
the largest quadrilateral were independently found in 2018
by Vahideh Keikha (personal communication, manuscript in preparation)
and by Kai Jin (personal communication, manuscript previously
submitted to a conference), and they are essentially the same as the
algorithm given here.  According to \cite{bddg-fep-85}, a linear-time
solution is given in unpublished notes of Michael Shamos from
1974~\cite{s-picg-74}.  Given that the solution is so simple, this is
plausible, but I have not been able to confirm it.

While the algorithms that we develop were known,
the observation that the two problems are so closely connected
(\Cref{lem:optimality-conjugate})
appears to be new.
A similar dual connection between the anchored versions of two
problems
exists between the
largest contained and the smallest enclosing \emph{triangle}.
This connection
was first
noted and exploited in the linear-time algorithm of
%pioneered by
Chandran and Mount~\cite{cm-paeet-92} for these problems, see also
%class notes~
\cite[Lemmas 4.i and 14]{r-class} for a slightly more stringent treatment in the
style of \Cref{lem:optimality-conjugate}.

\section{Conjugate Pairs}

A \emph{direction} is given by a nonzero vector
$\uu\in \mathbb R^2$. Parallel vectors represent the same direction,
and opposite directions are considered equal.
 Directions are conveniently parameterized by the polar angle $\theta$:
$\uu (\theta)=\binom{\cos\theta}{\sin\theta}$.

We denote the quadrilateral contained in $P$ conventionally by
its four corners $ABCD$. For the parallelogram that surrounds $P$, it will
be better to
denote it by the four \emph{sides} $abcd$, leaving
the corners anonymous, see \Cref{fig:conjugate}.

\begin{definition}
  \begin{enumerate}[a)]
  \item
     A quadrilateral $ABCD$ is \emph{D-anchored} to $\uu$ if
     the diagonal $AC$ is parallel to~$\uu$.%
   \item A parallelogram $abcd$ %with sides $a,b,c,d$
     is
     \emph{S-anchored} to $\uu$ if the two sides $b$ and $d$ are
     parallel to $\uu$.
%      (S stands for ``side''.)
  \end{enumerate}
\end{definition}
The letter D stands for ``diagonal'', and S stands for ``side''. We will sometimes just say
``\emph{anchored}'' if it is clear from the context which version we mean.
\begin{definition}
  Let $F=ABCD$
  be a quadrilateral, % contained in a convex region $P$
  and
  let $G=abcd$
  be a parallelogram. % containing $P$.
  We say that
  $F$
  and  $G$ are \emph{conjugate},
or $(F,G) $ form a \emph{conjugate pair},
  if
  \begin{compactenum}%[\indent a)]
  \item the diagonal $AC$ is parallel to the sides $b$ and $d$, and
\item each corner $A,B,C,D$ of $F$ lies on the corresponding {side}
 $a,b,c,d$  of $G$.
%  \item $A$ lies on the side $a$,
%  \item $B$ lies on the side $b$,
%  \item $C$ lies on the side $c$, and
%  \item $D$ lies on the side $d$.
  \end{compactenum}
\end{definition}
Because of the first condition, the two elements $F$ and $G$ of a
{conjugate pair} are % necessarily
anchored to the same direction.  Because of
the % remaining
second condition, % (b), %--e),
%requires that the points of $F$ lie on the sides of $G$ and not just on the
%lines defining these sides. Therefore,
$F$ is a convex quadrilateral contained in $G$. It is possible that
$F$ degenerates to a triangle because it is not necessarily
strictly convex, and it may even happen that some corners coincide.

\begin{figure}[htb]
  \centering
  \includegraphics{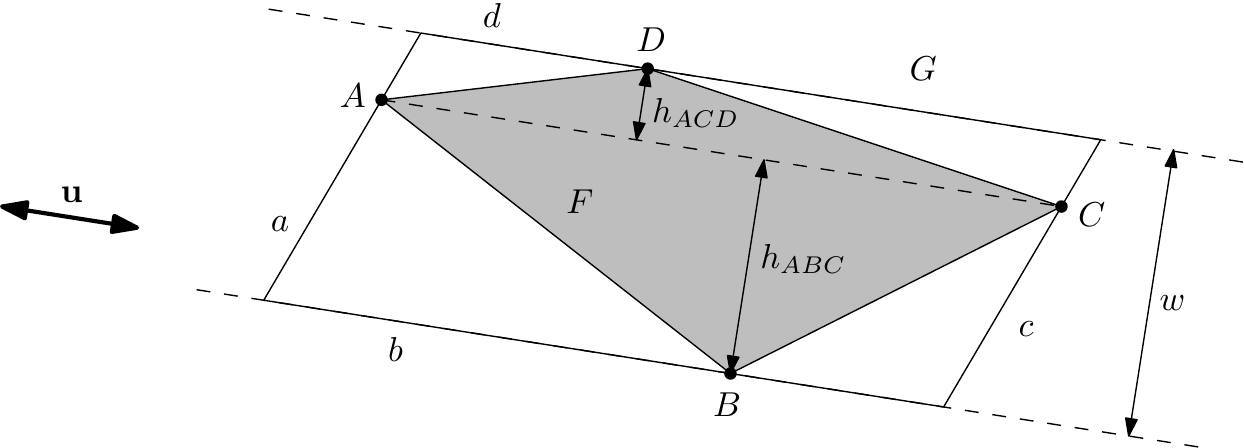}
  \caption{A conjugate pair $(F,G)$. The quadrilateral
$F=ABCD$ is D-anchored and the parallelogram $G=abcd$ is S-anchored
    to the direction $\uu$.
    The heights
$h_{ACD}$ and $h_{ABC}$ of the two triangles into
which $ABCD$ is decomposed by the diagonal $AC$ sum up to
%the width
the distance $w$ between the lines through $b$ and $d$. %_P(\uu^\perp)=h_{ACD}+h_{ABC}$.
  }
  \label{fig:conjugate}
\end{figure}

% \begin{figure}[htb]
%   \centering
%   \includegraphics{4gon}
%   \caption{A conjugate pair $(ABCD,abcd)$, anchored to the direction $\uu$  }
%   \label{fig:example}
% \end{figure}

The following basic geometric lemma considers
 a {conjugate pair}  $(F,G) $ in isolation %, without reference to the
 % region~$P$,
 and proves some optimality properties of $F$ and $G$ with
 respect to each other.

\begin{lemma}
% Let  $(F,G) $ be a {conjugate pair}, anchored to some direction $\uu$.
 \begin{enumerate}[ a)]
 \item \label{simple:largest}
 Let  $G$ be a \para, S-anchored to some direction
 $\uu$.
 Then a \qu\ $F$ that is contained in $G$
   and is D-\anc\ to $\uu$ 
 is a {largest} quadrilateral with these properties
if and only if
  $(F,G)$ is a conjugate pair.
 \item \label{simple:smallest}
 Let  $F$ be a \qu, D-anchored to some direction
 $\uu$.
Then a \para\ $G$ that contains $F$
   and is S-\anc\ to $\uu$
 is a {smallest} parallelogram with these properties
 if and only if
  $(F,G)$ is a conjugate pair.
 \item \label{area}
If $(F,G)$ is a conjugate pair,
 the area of $G$ is twice the area of $F$.
 \end{enumerate}
\label{lem:simple}
\end{lemma}

% \begin{lemma}\label{lem:simple}
%  Let  $(F,G) $ be a {conjugate pair}, anchored to some direction
%  $\uu$.
%  \begin{enumerate}[ a)]
%  \item \label{simple:largest}
%  Then $F$ is a {largest} quadrilateral that is contained in $G$
%    and is D-\anc\ to $\uu$.
%  \item \label{simple:smallest} $G$ is a {smallest} parallelogram that contains $F$
%    and is S-\anc\ to $\uu$.
%  \item \label{necessary}
% Moreover, in both cases, o
%  \item \label{area} The area of $G$ is twice the area of $F$.
%  \end{enumerate}
% \end{lemma}
\begin{proof}
  See \Cref{fig:conjugate}.  Since $F$ and $G$ are required to be \anc\ to the
  same direction, the first condition for a conjugate pair is always
  satisfied. The question is whether the four sides of $G$ are
  incident to the four corresponding corners of~$F$.

Let $|b|=|d|$ denote the length of the two sides of $G$ that are
parallel to $\uu$. Then, given that
the diagonal $AC$ should be parallel to $\uu$
% $F$ should be D-\anc\ to $\uu$
and contained in $G$, it is
clear that
\begin{displaymath}
|AC|\le
  |b| =
  |d|,  
\end{displaymath}
with equality if and only if the sides $a$ and $c$ touch $A$ and $C$.

Moreover, if $F$ is contained in $G$, the distance between $B$ and
$D$, when projected to the direction perpendicular to~$\uu$, is at
most the distance $w$ between the lines through $b$ and $d$:
 \begin{displaymath}
   \bigl|\uu^\perp \cdot (D-B)\bigr| \le w,
 \end{displaymath}
with equality if and only if the sides $b$ and $d$ touch $B$ and $D$.

 \eqref{simple:largest}
The \qu\ $F=ABCD$ is composed of
 the triangles $ABC$ and
$ACD$, which share the common base $AC$. Therefore,
the area of  $F$ is expressed in terms of the
heights $h_{ABC}$ and $h_{ACD}$ of these triangles as
   \begin{equation}
     \label{eq:area}
\tfrac 12 |AC|\times (h_{ABC}+h_{ACD})
=\tfrac 12 |AC|\times \bigl|\uu^\perp \cdot (D-B)\bigr|
\le 
\tfrac 12 |b| w,
   \end{equation}
   and we have just seen that equality holds if and only if the four
   sides of $G$ touch the corresponding corners of $F$. This
   proves~\eqref{simple:largest}.
 %  for {any} \emph{fixed} segment $AC$ parallel to $\uu$, the area of $ABCD$ is
% largest if the distance of $B$ and $D$ from the line $AC$ is
% maximized,
% because these distances are the heights of
%  the triangles $ABC$ and
% $ACD$ over their common base $AC$. Therefore is it clear that $B$ and
% $D$ touch the sides $b$ and $d$,
%  and the area of $ABCD$ is
%    \begin{equation}
%      \label{eq:area}
% \frac 12 |AC|\cdot w
%    \end{equation}
% where $w$ is the distance between the lines $b$ and $d$.
% %\eqref{eq:area}. % $\frac 12 |AC|\cdot w_P(\uu^\perp)$. 
% The only quantity that is left to maximize in this
% formula is
% $|AC|$, and this segment is longest if and only if $A$ and $C$ touch the sides $a$ and $c$.
% As long as $AC$ is parallel to $\uu$ and the endpoints $A$ and $C$
% move on the lines $a$ and $c$, the length $|AC|$ is constant, and it
% equals the length of the sides $b$ and $d$.
The area of the parallelogram $G$ is
\begin{equation} \label{eq.area-parallelogram}
  |b| w =
  |d| w,\end{equation}
which equals twice the area
of $F$
 in~% the left-hand side of 
\eqref{eq:area}, % $|AC|\cdot w$,
and this proves~\eqref{area}.

To prove
\eqref{simple:smallest},
we use 
\eqref{eq:area} in the other direction,
giving a lower bound
on
the area \eqref{eq.area-parallelogram} of any \anc\ parallelogram $G$
containing $F$. % in terms of $F$.
Again, since equality in
\eqref{eq:area} holds if and only if $F$ and $G$ are conjugate,
\eqref{simple:smallest} has been proved.
%The distance $w$ 
% between the parallel lines $b$ and $d$ is smallest if these lines
% touch $F$.
\end{proof}

%   // follows: the diagonal vector $r-t$ is equal to the vector
%of the two sides $B-A$ and $C-D$,

\begin{figure}[htb]
  \centering
  \includegraphics{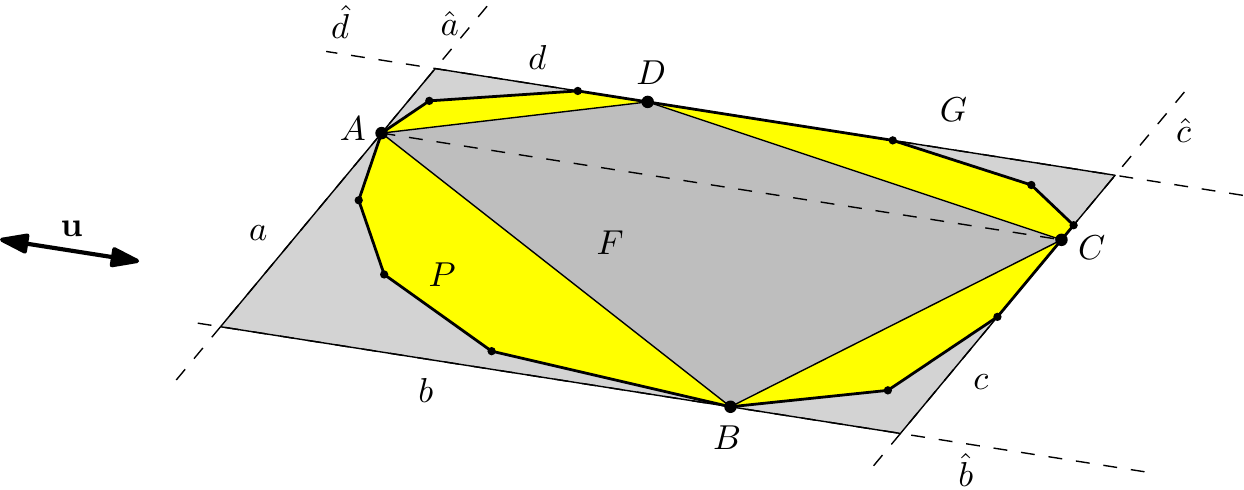}
  \caption{A conjugate pair $(F,G)=(ABCD,abcd)$ anchored to the direction
    $\uu$ and sandwiching a convex polygon $P$.
 }
  \label{fig:4gon}
\end{figure}

The following crucial lemma gives the optimality condition for the
\anc\ versions of the two problems.

\begin{lemma}[Characterization of Optimality by Conjugate Pairs]
\label{lem:optimality-conjugate}
Let $P$ be a convex polygon in the plane, and let $\uu$ be a direction.
  \begin{enumerate}[ a)]
  \item A quadrilateral $F$ %\subseteq P$ % anchored to $\uu$
that is
    D-\anc\ to
    $\uu$ 
and contained in $P$ 
    is a {largest} quadrilateral with these properties
if and only if there is a parallelogram containing $P$ that
    is conjugate to $F$.
\item 
  A parallelogram $G$ %\supseteq P$
 enclosing $P$ and S-anchored to $\uu$
  is
  a
{smallest} %enclosing
 parallelogram % S-\anc\ to $\uu$
 with these properties
 if and only if there is 
    a quadrilateral % $F$
    contained
 in $P$ that is conjugate to~$G$.    
  \end{enumerate}
\end{lemma}
\goodbreak
\begin{proof}[Proof of sufficiency.]
In both cases, there is a conjugate pair $(F,G)$ such that
the convex region $P$ is sandwiched between them:
$F\subseteq P \subseteq G$.

a) By
\Cref{lem:simple}\ref{simple:largest}, $F$ is even the largest D-\anc\
\qu\ inside the larger region $G\supseteq Q$. Thus, there cannot be a
larger \anc\ \qu\ in~$P$.

b) By
\Cref{lem:simple}\ref{simple:smallest}, $G$ is even the smallest S-\anc\
\para\ that encloses the smaller region $F\subseteq Q$. Thus, there cannot be a
smaller \anc\ \para\ enclosing~$P$.
\end{proof}
Necessity of the conditionsis not needed for the correctness of our
algorithm, and it will only be proved later
as an easy consequence of \Cref{lemma-AC-BD}, see page~\pageref{necessity}.
Alternatively, there are easy direct proofs
(cf.~\cite[Lemma~2]{stwe-fmep-94}), even for arbitrary convex regions.

The lemma is also a manifestation of
linear programming duality, since the problem of finding the longest
chord $AC$ with a given
direction can be formulated as a linear program.

% From the necessity of the optimality conditions, which we have not
% proved yet, it would follow that there
% conjugate pairs \anc\ to any direction, simply because there are
% largest contained quadrilaterals and smallest enclosing parallelograms 
%  \anc\ to any direction.
% We will prove the existence of
% conjugate pairs for any direction directly, by an algorithm that constructs

% The conditions are also necessary, but we don't need this for our algorithm.

\section{Constructing all Conjugate Pairs in Linear Time}
\label{sec:conjugate}
The idea is to construct
% the largest anchored \qu s
conjugate pairs
$(F(\theta),G(\theta))$ with $F(\theta)\subseteq P\subseteq G(\theta),$
%
%$%ABCD=
%(A(\theta)B(\theta)C(\theta)D(\theta),
%%ABCD=
%a(\theta)b(\theta)c(\theta)d(\theta))$
for all directions $\uu(\theta) $ in the range
$0^\circ\le\theta \le 180^\circ$.  By the sufficient criterion of
\Cref{lem:optimality-conjugate}, these are largest \anc\ contained
\qu s and smallest \anc\ enclosing \para s. Hence, the overall
largest contained \qu s and smallest enclosing \para s will be among
them.  \looseness-1

The following straightforward observation separates the task of finding an \anc\ conjugate
pair $(F,G)$ into two subtasks.
The first
 involves $A$, $C$, $a$, and $c$, and it is concerned with the
direction %$\theta$
 of the diagonal $AC$.
The other task
involves $B$, $D$, $b$, and $d$, and it is concerned with the
direction %$\phi$
 of the sides $b$ and $d$.
%REIENFOLGE?
 A pair of points on the boundary of a convex region $P$ that admits parallel supporting lines is
called \emph{antipodal}.
% A pair of points that does admit parallel supporting lines will be
% called antipodal. p.178--179 \cite{PS}

\begin{lemma}\label{largest-anchored}
  Let $P$ be a convex region in the plane and $\uu$ be a direction.
  A conjugate pair
  $(ABCD,abcd)$
with $ABCD\subseteq P\subseteq abcd$
     and {anchored} to $\uu$ is found as follows, see \Cref{fig:4gon}.
Here, the \para\ $abcd$ is defined by two pairs of parallel lines $\hat a,\hat
c$ and $\hat b,\hat d$:
     \begin{enumerate}[\indent a)]
     \item \label{AC}
$AC$ is an antipodal pair of $P$ parallel to
       $\uu$, with supporting lines $\hat a$ and $\hat c$,
     \item  \label{BD}
$\hat b$ and $\hat d$ are the two opposite lines of support parallel
       to $\uu$, and
$B$ and $D$ are points where these lines touch $P$.
\textup(Thus, $BD$ is also an antipodal pair.\textup)
 \qed
\end{enumerate}
% Every quadrilateral % $ABCD$
%  satisfying these conditions is optimal, and
%    its area is % of $ABCD$ is
%    \begin{equation}
%      \label{eq:area}
% \frac 12 |AC|\cdot w_P(\uu^\perp).     
%    \end{equation}
\end{lemma}

% The idea is to construct
% % the largest anchored \qu s
% conjugate pairs
% $%ABCD=
% (A(\theta)B(\theta)C(\theta)D(\theta),
% %ABCD=
% a(\theta)b(\theta)c(\theta)d(\theta))$
% for all directions
% $\uu(\theta)
% $ in the range
%  $0^\circ\le\theta
% \le 180^\circ$ and select the largest among them.
% By \Cref{largest-anchored}, this task can be separated into two subtasks.
% Finding point pairs $B(\theta),D(\theta)$ that
% satisfy condition \eqref{BD}
% and finding longest contained segments $A(\theta)C(\theta)$, for all
% directions
% $\uu(\theta)$.
As stated in the following lemma, whose proof will be given in \Cref{sec:calipers},
 both tasks can easily be carried out with the
classical rotating-calipers technique.
%So far, the region $P$ was any convex region in the plane.
%Now
 We assume that $P$ is a convex polygon,
% $P%=(p_1,p_2,\ldots,p_n)
%$ is 
given by the ordered list of
its $n$ vertices. % in counterclockwise order.

% \begin{lemma}\label{lemma-BD}
%   In $O(n)$ time, one can find
%  direction angles
%  $0^\circ=\theta_0<\theta_{1}<\dots<
% \theta_{i-1}<\theta_{i}<\dots<\theta_k<\theta_{k+1}=180^\circ$, 
% and two sequences of 
% \textup(not necessarily distinct\textup)
% vertices $B_1,B_2,\ldots,B_{k+1}$
% and $D_1,D_2,\ldots,D_{k+1}$,
% with
%  $k\le n$,
%  such that
%  in each
% closed interval $[\theta_{i-1}\ldt \theta_{i}]$,
% the points $B_i$ and $D_i$ 
% satisfy condition (\ref{BD})
% of \Cref{largest-anchored} for $\uu(\theta)$.
% \end{lemma}

% \begin{lemma}\label{lemma-AC}
%   In $O(n)$ time, one can find
%  direction angles
%  $0^\circ=\theta'_0<\theta'_{1}<\dots<
% \theta'_{i-1}<\theta'_{i}<\dots<\theta'_n<\theta'_{n+1}
% =180^\circ$, 
% and a sequence of vertex-edge pairs
% $(Q_1,e_1),(Q_2,e_2),\dots,(Q_{n+1},e_{n+1})$,
% %with $\ell \le n$,
%  such that for any~$\theta$
%  in each
% closed interval $[\theta_{i-1}'\ldt \theta'_{i}]$,
%  a longest segment
% $A(\theta)C(\theta)$ contained in $P$ parallel to
%        $\uu(\theta)
% $
% can be found by intersecting the line through
% $Q_i$
%  parallel to
%        $\uu(\theta)$ with the edge $e_i$.
% \end{lemma}

\begin{lemma}\label{lemma-AC-BD}
\begin{enumerate}[a)]
\item 
\label{lemma:AC}
  In $O(n)$ time, one can find a sequence of
 direction angles
 $0^\circ=\theta_0<\theta_{1}<\dots<
\theta_{i-1}<\theta_{i}<\dots<\theta_n<\theta_{n+1}
=180^\circ$,
and a corresponding sequence of vertex-edge pairs
$(Q_1,e_1),(Q_2,e_2),\dots,(Q_{n+1},e_{n+1})$,
%with $\ell \le n$,
such that
for any~$\theta$
 in each
closed interval $[\theta_{i-1}\ldt \theta_{i}]$,
an antipodal segment
$A(\theta)C(\theta)$ parallel to
       $\uu(\theta)
$
can be found by intersecting the line through
$Q_i$
 parallel to
       $\uu(\theta)$ with the edge $e_i$.
 The lines parallel to $e_i$ through 
$Q_i$ and $e_i$
 are the corresponding supporting lines.
\item 
\label{lemma:BD}
  In $O(n)$ time, one can find a sequence of %not necessarily distinct
 direction angles
 $0^\circ=\phi_0<\phi_{1}<\dots<
\phi_{i-1}<\phi_{i}<\dots<\phi_k<\phi_{k+1}=180^\circ$, 
and a corresponding sequence of antipodal
 pairs of vertices
$(B_1,D_1),(B_2,D_2),\allowbreak\ldots,\allowbreak(B_{k+1},D_{k+1})$,
with
 $k\le n$,
 such that for any $\phi$
 in each
closed interval $[\phi_{i-1}\ldt \phi_{i}]$,
%the points $B_i$ and $D_i$  form an antipodal pair.
the lines through
$B_i$ and $D_i$ parallel to the direction $\uu(\phi)$ are
 supporting lines.
% $\phi'_0<\phi'_{1}<\dots<
%\phi'_{i-1}<\phi'_{i}<\dots<\phi'_n<\phi'_{n+1}
%=180^\circ+\phi'_0$, 
\end{enumerate}
\label{lemma-AC-BD}
\end{lemma}

We remark that the sequence $(B_i,D_i)$ % in part~\eqref{lemma:BD}
does not necessarily include \emph{every}
 pair of antipodal vertices:  For each pair of opposite parallel
edges of $P$, there are two pairs of antipodal vertices which admit
parallel supporting lines of only one direction. These pairs don't appear in the list.

 It is now clear how to proceed
with the help of \Cref{lemma-AC-BD}.
Since the areas of a conjugate pair are related
by \Cref{lem:simple}\ref{area},
let us
ignore the enclosing
\para s
$a(\theta)b(\theta)c(\theta)d(\theta)$
 and 
 concentrate on the inner \qu s
$A(\theta)B(\theta)C(\theta)D(\theta)$.
 We merge the lists
of breakpoints $\theta_0,\theta_1,\ldots$
and $\phi_0,\phi_1,\ldots$ and obtain a list of $O(n)$
intervals such that in each interval, there are largest \anc\ \qu s
$%ABCD=
A(\theta)B(\theta)C(\theta)D(\theta)$ with a fixed structure:
The points $B(\theta)=B$ and $D(\theta)=D$ are fixed vertices.
On the diagonal
$A(\theta)C(\theta)$, one point, say
$A(\theta)=A$, is fixed to a vertex $Q_i$, while the other point
$C(\theta)$ moves on a fixed edge $e_i$.

In a quadrilateral
$ABC(\theta)D$
with one moving point $C$, the area is a linear function of $C$.
As $\theta$ increases,
the corner $C(\theta)$ moves monotonically on some edge $e_i$, and 
therefore, the extremes are attained at the endpoints of the interval.
We thus just need to evaluate the area at all interval endpoints
$\theta_i$ and
$\phi_i$ of the merged sequence and pick the largest or smallest one.
Since each endpoint belongs to two intervals, the 
\qu\ $%ABCD=
A(\theta)B(\theta)C(\theta)D(\theta)$ prescribed by
\Cref{lemma-AC-BD%,lemma-AC
} may be ambiguous,
but this does not matter. All these \qu s have the
same area.

\begin{theorem}
  \begin{enumerate}[ a)]
  \item 
  The 
 quadrilateral of largest area contained
in a
  convex polygon can be found in linear time.
\item 
  The 
\para\ of smallest area enclosing a
  convex polygon can be found in linear time.
\qed
  \end{enumerate}
\end{theorem}
 Pseudocode for the algorithm is given in
 \Cref{sec:algorithm-pseudocode},
 and
prototype implementations of the algorithms in 
 \Cref{sec:algorithm-pseudocode,sec:4gon,sec:schwarz}
in \textsc{Python}
 are contained in the source files
 of this preprint.
\section{Discussion}
It is perhaps instructive to reflect on some features of this algorithm and compare it
to other approaches.  An easy property of largest \qu s (in
fact, largest $k$-gons for any~$k$) contained in a polygon~$P$ is
the \emph{vertex property}:
Their corners must be vertices of~$P$. 
 Our algorithm does not use this property at all.
It considers an
%a-priori
 infinite family $%ABCD=
A(\theta)B(\theta)C(\theta)D(\theta)$ of \qu s.
Even after reducing them to a discrete set of directions
(the interval endpoints
$\theta_i$ and~$\phi_i$), many of these candidates don't fulfill the vertex
property.
Most previous algorithms for largest contained $k$-gons, and in particular,
the algorithms of Dobkin and
Snyder~\cite{dobkin79}, consider only $k$-gons with the vertex property.
By concentrating on the vertex property too early, one may miss
useful avenues to finding good and simple algorithms.

We may of course still use the vertex property as an ``afterthought'' to introduce shortcuts and
simplify the algorithm.  
For example, once the
point $C(\theta)$ lies in the middle of an edge, one can skip the area
computations and fast-forward $\theta$ until $C(\theta)$ arrives at a
vertex.
% (This shortcut requires some care in the case of parallel edges.)
(For the problem of the largest contained triangle, the analogous step is described
in \cite[Section~8]{r-class}.)

There are other possible simplifications.
 The two lists of breakpoints $\theta_0,\theta_1,\ldots$
and $\phi_0,\phi_1,\ldots$ need not be computed separately in
advance. They can be generated on the fly
as they are processed,
%while merging
%them, % and to process them right away,
after an
appropriate initialization.
We have described the algorithm in terms of angles $\theta$ for
convenience. When implementing the algorithm on a computer, it is better to
avoid angle calculations and
use % carry out the elementary steps by
 direct comparisons of vector directions or signed areas,
see \Cref{sec:primitive}.
(Anyway, since the problem is invariant under affine transformations,
angular  quantities are not really suited to the problem.)

In \Cref{sec:4gon}, we show the whole simplified algorithm for the
largest contained quadrilateral. This algorithm is actually so simple
that one can as well derive it directly from the property that
$AC$ must form an antipodal vertex pair, without going through the
continuous family
$A(\theta)B(\theta)C(\theta)D(\theta)$.
The same remark holds for the
smallest enclosing \para.
\Cref{sec:schwarz} shows a variation of the algorithm
following~\cite{stwe-fmep-94}
 that is just as simple.

\section{Rotating Calipers}
\label{sec:calipers}

% A pair of points that does admit parallel supporting lines will be
% called antipodal. p.178--179 \cite{PS}
\begin{proof}[Proof of \Cref{lemma-AC-BD}]
For part \eqref{lemma:AC}, we need antipodal points for all directions.
An algorithm for listing
all antipodal pairs of \emph{vertices} of a convex polygon $P$ is
given in
\cite[Section~4.2.3]{PS}.
We just need to ``fill the gaps'' in order to
get antipodal pairs for a continuous range of directions

    Let $f%(\phi)
$ and $g%(\phi)
$ be two opposite lines of support
in direction $\uu(\phi)$, see \Cref{fig:movement}. We will increase %continuously 
 $\phi$
from $\phi=0^\circ$ to $\phi= 180^\circ$ and maintain the points $A$ and
$C$ where they touch $P$. Since we want these points to move
continuously,
% jump at certain directions, 
we parameterize the process by a new parameter $t=\phi+s$, where $s$ it the combined
distance moved by $A(t)$ and $C(t)$ along the boundary of $P$ since
the beginning.
We start with $A(0)$ and $C(0)$ as the lowest and highest points of
$P$.
In case of ties, we take the leftmost lowest and the rightmost highest
point.
\Cref{fig:movement}a shows $\phi$ and the distances
 $s_A$ and $s_C$ moved by $A$ and $C$, from which $s$ is computed as $s=s_A+s_C$.

\begin{figure}[htb]
  \centering
  \includegraphics{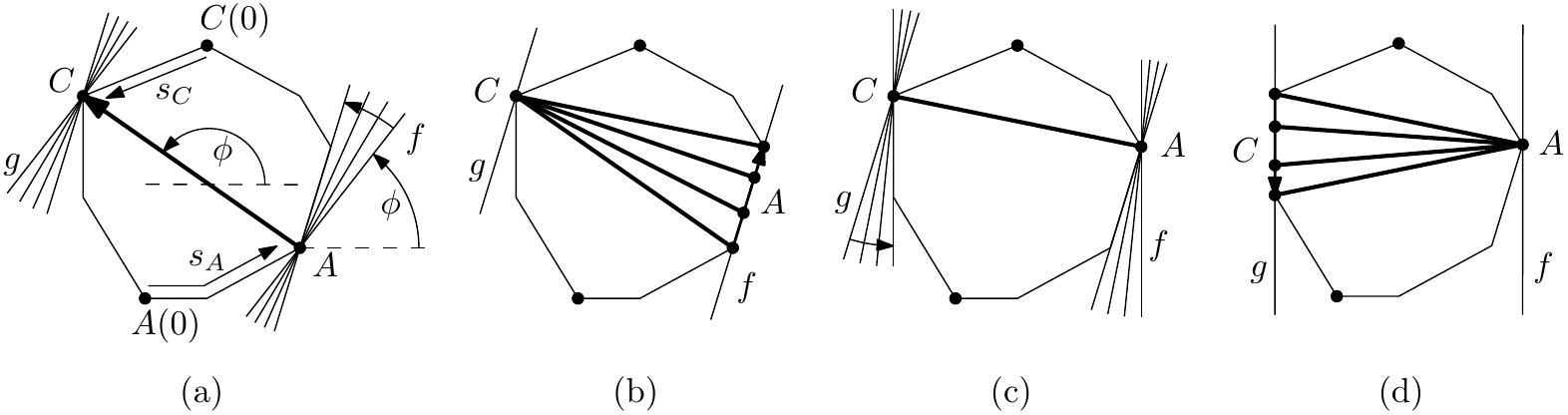}
  \caption{Four successive stages of the circular sweep: (a) The antipodal points $A=A(t)$ and $C=C(t)$
together with the parallel support lines $f=f(t)$ and $g=g(t)$,
for the parameter $t=\phi+s_A+s_C$.
%s the parameters $\phi$ and $s=s_A+s_C$; 
The angle
$\phi$ increases until $f$ or $g$ hits an edge.
(b) The line $f$ has hit an edge. $C$ is stationary and $A$ slides along
this edge.
(c)~$\phi$~increases further, and $g$ hits an edge.
(d)~$A$ is stationary and $C$ moves.}
  \label{fig:movement}
\end{figure}

Now we start to increase $t$.
Whenever $\uu(\phi)$ is parallel to an edge of $P$, we continuously
advance $A(t)$ or $C(t)$ to the other endpoint of this edge,
increasing $s$ while leaving $\phi$ constant.
If $P$ has two sides
parallel to $\uu(\phi)$, we arbitrarily use the convention that we first
advance $A(t)$ and then $C(t)$.
Now, $f(t)$ and $g(t)$ are ready to tilt around the vertices $A(t)$ and $C(t)$,
increasing $\phi$ while $s$ remains constant, until
$f(t)$ or $g(t)$ hits the next edge.

We continue this process in a loop until $\phi=180^\circ$.  At this
point, $A$ and $C$ have swapped places, and $s$ equals the perimeter
of $P$.  The segment $A(t)C(t)$ has completed a % counterclockwise
rotation by $180^\circ$.

The points $A(t)$ and $C(t)$ move continuously in counterclockwise
direction as a function of~$t$, and
for every $t$,
the points $A(t)$ and $C(t)$ are antipodal,
as witnessed by supporting lines $f(t)$ and $g(t)$.
Thus we have achieved our primary goal of finding an antipodal pair
for every direction.

The parameter range of $t$ is decomposed into intervals where
$A$ remains stationary,
$C$ remains stationary, or both points remain stationary.
We cut out those intervals where none of the points move.
For the remaining intervals, we choose yet another %a different
parameterization,
namely by the direction
$\uu(\theta)$ pointing from $A(t)$ to $C(t)$.

Each of the remaining intervals is characterized by one stationary
point, $Q_i$, while the other point moves on a fixed edge, $e_i$.
If $\uu(\theta)$ is the direction pointing from $A(t)$ to $C(t)$,
The breakpoints $\theta_{i-1}$ and 
$\theta_i$
 are the directions at the end of the intervals, 
when both $A(t)$ and $C(t)$ are at vertices.
It only remains to rearrange the interval breakpoints cyclically modulo
$180^\circ$ in order to start with $\theta_0=0^\circ$.
Since each interval advances either $A$ or $C$ by one vertex and
 $A$ and $C$ together make a full tour around $P$, the number of
interval breakpoints
$\theta_i$
 is~$n$.

\smallskip

 Part \eqref{lemma:BD} of the lemma is straightforward. In fact, it can be
 obtained by the same circular sweep as above, with the
 straightforward parameterization by the angle $\phi$,
 concentrating only on the points $A(\phi)$ and $C(\phi)$
 where the % right and left
 supporting lines in direction
 $\uu(\phi)$ touch $P$.
 (These points will take
 the roles of $B_i$ and $D_i$ in the lemma.)

% .....

% For
%  $0^\circ\le\theta
% \le 180^\circ$, we let $B$ and $D$ be the vertices
%  where the right and
%  left supporting lines in direction
%  $\uu(\theta)$ touch $P$.
% These vertices are unique whenever
%  $\uu(\theta)$ is not
% is parallel to some edge % $p_{j-1}p_j$ 
% of $P$, but they jump at these directions.

The breakpoint directions $\phi_i$ are therefore the directions
where $A(\phi)$ or $C(\phi)$ jumps.
These are the directions
 for which
  $\uu(\phi_i)$ is parallel to some edge % $p_{j-1}p_j$ 
of $P$.
There are at most $n$ such angles.
 The sequence $\phi_1,\phi_2,\dots$ is obtained by
merging the two lists of edge directions obtained from traversing the
left boundary of $P$ and the right boundary of $P$ counterclockwise,
between the extreme points in vertical direction.
\end{proof}

\begin{proof}[Proof of necessity in
\Cref{lem:optimality-conjugate}.]\label{necessity}

a) Assume that $F'$ is a largest
quadrilateral
that is
    D-\anc\ to
    $\uu$ 
and contained in $P$.
\Cref{lemma-AC-BD} together with
\Cref{largest-anchored}
implies that, for this direction $\uu$,
there exists an \anc\ conjugate pair $(F,G)$
with $F\subseteq P \subseteq G$.
 By the sufficiency part of
\Cref{lem:optimality-conjugate}, which has already been proved,
$F$ is a largest \anc\ \qu\ contained in $P$,
and therefore of the same area as~$F'$.
 By
\Cref{lem:simple}\ref{simple:largest},
$F$ is even a largest \anc\ \qu\ contained in the larger area~$G$.
By the necessity statement in the same lemma,
 since $F'$ is also contained in $G$,
%and no \anc\ \qu\ $F'$ contained in $G$ 
$F'$ can only have the same area as $F$
 if it forms a conjugate pair $(F',G)$ with $G$. This proves the
 necessity for Part~(a).
The proof of Part (b) is completely analogous.
\end{proof}

\bibliographystyle{plainurl}
\bibliography{./b}

\goodbreak
\appendix
\section{The Algorithm in Pseudocode}
\label{sec:algorithm-pseudocode}

For completeness, we give the pseudocode for
our algorithm.
% computing the area
%of the largest  4-gon contained in a convex polygon
%and the smallest \para\ enclosing a convex polygon
We assume that the convex polygon
 $P=(p_1,p_2,\ldots,p_n)
$ is given by the ordered list of
its $n$ vertices in counterclockwise order.
We assume that $n\ge 3$, and
we look for a largest contained \qu\ $ABCD=p_ap_bp_cp_d$ in 
counterclockwise order, and a smallest enclosing \para,
also in counterclockwise order.
Indices of polygon vertices are considered modulo~$n$.

In contrast to the algorithm that is sketched in
\Cref{sec:conjugate}, we don't start with $\theta_0=0^\circ$, but
we start more conveniently with the antipodal pair defined by $A=p_1$ and the point $C$
opposite to the edge $p_1p_2$.

\subsection{Primitive Operations}
\label{sec:primitive}

The basic
predicate of this algorithm is
a comparison between two directions
$\uu=\binom {x_1}{y_1}$ and $\mathbf{v}=\binom {x_2}{y_2}$, which 
can be calculated
with only two multiplications
 as the sign of a $2\times2$ determinant that expresses
the signed area of the parallelogram spanned by $\uu$ and $\mathbf{v}$:
\begin{displaymath}
  \det\left(
\mathbf{u},\mathbf{v}
%\binom {x_1}{y_1},
%\binom {x_2}{y_2}
\right) :=
  \left|
    \begin{matrix}
      x_1&x_2\\
      y_1&y_2\\
    \end{matrix}\right|
  = x_1y_2-x_2y_1
=-  \det\left(
\mathbf{v},\mathbf{u}
\right)
\end{displaymath}
This is positive if $\mathbf{v}$ lies counterclockwise from $\uu$.
The area of a %counterclockwise
quadrilateral $ABCD$ is
$\pm\frac12\det\bigl((C-A),(D-B)\bigr)$.

Frequently, the algorithm makes comparisons between triangle areas over a common
basis. This should also be
calculated as a
$2\times2$ determinant.
For example,
$\area p_a p_{a+1}p_{c+1}-
\area p_a p_{a+1}p_{c} = \frac12\det\bigl((p_{a+1}-p_a),(p_{c+1}-p_c)\bigr)$
if the two triangles are oriented counterclockwise.
Since we find the formulation involving triangle areas
geometrically more appealing, we have not replaced it in our pseudocode.

\subsection{Largest and smallest anchored \qu s}

\begin{lemma}[{\cite[Lemma~1]{stwe-fmep-94}}]
\label{2-flush}
%Let $P$ be 
There is a smallest enclosing \para\ $abcd$ such that
\begin{compactenum}
\item 
 at least one of the sides $a$ and $c$ touches an
edge of $P$, and
\item 
 at least one of the sides $b$ and $d$ touches an
edge of $P$.  
\end{compactenum}
\end{lemma}
\begin{proof}
A smallest enclosing \para\ $abcd$ must be a smallest enclosing
 \para\ \emph{\anc} to the direction of $b$ and $d$, and hence there
 must be a conjugate pair $(ABCD,abcd)$, see \Cref{fig:4gon}.
If the side $a$ or $c$ doesn't already touch an edge of $P$, these
sides can be tilted around $A$ and $C$ without changing the area,
until one of the sides hits an edge of $P$.

Afterwards, we can apply the same argument to the direction of $a$ and
$c$ and ensure that $b$ or $d$ touches an edge of $P$.
\end{proof}

As a consequence of part 2, when looking for
the smallest enclosing \para, it is sufficient to look at \para s that
are S-\anc\ to the directions of the edges of $P$.
%(One can also show that in \emph{every} smallest enclosing \para, at
%  least one side touches an edge of $P$.)
We have already mentioned that a largest contained \qu\ can be found
among those \qu s that use only vertices of $P$. Thus it is sufficient
to look at \anc\ \qu s for which $A$ and $C$ lie at vertices.
This explains the places where areas are compared against the current
minimum or maximum in the following program.

\goodbreak
\subsection{Pseudocode}
\label{sec:pseudocode}

\begin{tabbing}
%  \qquad\=\+
$a_0 := a := 1$\\
$c:= 2$\\
\textbf{while} $\area p_ap_{a+1}p_{c+1}>
\area p_ap_{a+1}p_{c}$:\\
  \qquad\=
 $c_0 := c := c+1$  (find the point $p_c$ with supporting line parallel to
 $p_ap_{a+1}$.)\\
$\textit{next\_AC} := \textrm{``A''}$
(The corner $A$ slides on the edge $p_ap_{a+1}$.)\\
$\uu^{\mathrm{AC}}
 := p_{c}-p_{a+1}$ (the direction $\uu$ where $A$ hits the next vertex)\\
$b := a$\\
\textbf{while} $\area p_cp_ap_{b+1}>
\area p_cp_ap_{b}$:\\
\>
 $b := b+1$  (find the point $p_b$ with supporting line parallel to
 $p_ap_c$.)\\
$d := c$\\
\textbf{while} $\area p_ap_cp_{d+1}>
\area p_ap_cp_{d}$:\\
  \qquad\=
 $d := d+1$ (find the other point $p_d$ with supporting line parallel to
 $p_ap_c$.)\\
%(Determine whether %the line
\textbf{if} $\area p_bp_{b+1}p_{d+1}\le
\area p_bp_{b+1}p_{d}$:\\
\>$\textit{next\_BD} := \textrm{``B''}$
(The \para\ side $b$ hits an edge of $P$ before $d$ does.)\\
\>
$\uu^{\mathrm{BD}}
 := p_{b+1}-p_{b}$\\
\textbf{else:}\\
\>$\textit{next\_BD} := \textrm{``D''}$
(The \para\ side $d$ hits an edge of $P$ before $b$ does.)\\
\>$\uu^{\mathrm{BD}} := p_{d}-p_{d+1}$\\
%$b_0 := b$\\
%$d_0 := d$\\
$\textit{maxarea} := 0$ (the area of the largest contained \qu)\\
$\textit{minarea} := \infty$ (the area of the smallest enclosing \para)\\
% $\textit{next\_AC} :=\textit{next\_BD} := \textrm{``unknown''}$\\
\textbf{repeat}\\
  \qquad\=\+
%\textbf{if} $\textit{next\_BD} = \textrm{``unknown''}$:\\
%  \qquad\=\+
%\\\<\textbf{if} $\textit{next\_AC} = \textrm{``unknown''}$:
%
\textbf{if} $\det(\uu^{\mathrm{BD}},\uu^{\mathrm{AC}})
\ge0$:\\
  \qquad\=\+
(The \para\ side
$b$ or $d$ touches an edge of $P$.)\\
%\<(i)\>
\textbf{if} $\textit{next\_AC} = \textrm{``A''}$:\\
\qquad\=\+
construct the point $A$ on the line $p_ap_{a+1}$ such that $p_cA$ is
parallel to $\uu^{\mathrm{BD}}$\\
\<\<\<($*$)\>\>\>
$\textit{minarea} := \min \{
\textit{minarea},
2\cdot\area Ap_bp_cp_{d}$\}
\-\\
\textbf{else}:\\
\>\+
construct the point $C$ on the line $p_cp_{c+1}$ such that $p_aC$ is
parallel to $\uu^{\mathrm{BD}}$\\
\<\<\<($**$)\>\>\>
$\textit{minarea} := \min \{
\textit{minarea},
2\cdot\area p_ap_bCp_{d}$\}
\-\\
%\<\<(ii)\>\>
\textbf{if} $\textit{next\_BD} = \textrm{``B''}$: \=$b := b+1$\\
\textbf{else}:
\>$d := d+1$\\
%$\textit{next\_BD} := \textrm{``unknown''}$\\
%\<\<(iii)\>\>
\textbf{if} $\area p_bp_{b+1}p_{d+1}\le
\area p_bp_{b+1}p_{d}$:\\
\qquad\=
$\textit{next\_BD} := \textrm{``B''}$
(The \para\ side $b$ hits an edge of $P$ before $d$ does.)\\
\>
$\uu^{\mathrm{BD}}
 := p_{b+1}-p_{b}$\\
\textbf{else:}\\
\>$\textit{next\_BD} := \textrm{``D''}$
(The \para\ side $d$ hits an edge of $P$ before $b$ does.)\\
\>$\uu^{\mathrm{BD}} := p_{d}-p_{d+1}$\\
\<\textbf{else}:
(The sliding corner $A$ or $C$ reaches a vertex of $P$.)\\
%\<\<(iv)\>\>
\textbf{if} $\textit{next\_AC} = \textrm{``A''}$: \=$a := a+1$\\
\textbf{else}:
\>$c := c+1$\\
$\textit{maxarea} := \max \{
\textit{maxarea},
\area p_ap_bp_cp_{d}$\}\\
%$\textit{next\_AC} := \textrm{``unknown''}$
%\<\<(v)\>\>
\textbf{if} $\area p_ap_{a+1}p_{c+1}\le
\area p_ap_{a+1}p_{c}$:
(Which of $A$ and $C$ slides on an edge of $P$?)\\
\qquad\=$\textit{next\_AC} := \textrm{``A''}$
(The corner $A$ slides on the edge $p_ap_{a+1}$.)
\\
\>$\uu^{\mathrm{AC}}
 := p_{c}-p_{a+1}$ (the direction $\uu$ where $A$ hits the next vertex)\\
\textbf{else:}\\
\>$\textit{next\_AC} := \textrm{``C''}$
(The corner $C$ slides on the edge $p_cp_{c+1}$.)
\\
\>$\uu^{\mathrm{AC}} := p_{c+1}-p_{a}$ (the direction $\uu$ where $C$ hits the next vertex)
\-\-\\
\textbf{until} $(a,c)=(c_0,a_0)$
\end{tabbing}
%The case distinction ($*$) is not necessary:
The area of $ Ap_bp_cp_{d}$ in line ($*$)
%$\area p_ap_bCp_{d}$ 
can be computed by the formula
% \begin{displaymath}
%  \pm  \frac12 \cdot \frac
% {
% \bigl(\uu^{\mathrm{AC}}\wedge (p_c-p_a)\bigr)
% \bigl(\uu^{\mathrm{BD}}\wedge (p_d-p_b)\bigr)
% }
% {\uu^{\mathrm{AC}}\wedge \uu^{\mathrm{BD}}}
% \end{displaymath}
% \begin{displaymath}
%  \pm  \frac12 \cdot \frac
% {
% \bigl|\uu^{\mathrm{AC}}, (p_c-p_a)\bigr| \cdot
% \bigl|\uu^{\mathrm{BD}}, (p_d-p_b)\bigr|
% }
% {\bigl|\uu^{\mathrm{AC}}, \uu^{\mathrm{BD}}\bigr|}
% \end{displaymath}
\begin{displaymath}
 \pm  \frac12 \cdot \frac
{\det
\bigl(p_{a+1}-p_a
, p_c-p_a\bigr) \cdot
\det\bigl(\uu^{\mathrm{BD}}, p_d-p_b\bigr)
}
{\det\bigl(p_{a+1}-p_a
, \uu^{\mathrm{BD}}\bigr)},
\end{displaymath}
and for % the analogous formula for
 the area of $p_ap_bCp_{d}$
in~($**$), we replace
$p_{a+1}-p_a$
by $p_{c+1}-p_c$ in two places. % in this formula.

\section{The Largest Contained 4-Gon}
\label{sec:4gon}

%For completeness, 
We give here the specialized algorithm for computing the area
of the largest  4-gon contained in a convex polygon $P$.
% %We assume that a convex polygon
%  $P=(p_1,p_2,\ldots,p_n)
% $ that is given by the ordered list of
% its $n$ vertices in counterclockwise order.
% We assume that $n\ge 5$, and
% we look for a largest contained \qu\ $ABCD=p_ap_bp_cp_d$ in 
% counterclockwise order.
%  Subscripts are incremented modulo~$n$.

In contrast to the algorithm that is sketched in
\Cref{sec:conjugate},
and also differently from \Cref{sec:pseudocode},
we start as in
\Cref{sec:calipers}
 with the points  $A(0)$ and $C(0)$ that have
horizontal supporting lines, 
see \Cref{fig:movement}.
\begin{tabbing}
  \qquad\=\+
Let $p_{a_0}$ be the leftmost vertex among the lowest vertices of $P$\\
Let $p_{c_0}$ be the rightmost vertex among the highest vertices of $P$\\
$a := b := a_0$\\
$c := d := c_0$\\
$\textit{maxarea} := 0$\\
\textbf{repeat}\\
  \qquad\=\+
\textbf{while} $\area p_cp_ap_{b+1}>
\area p_cp_ap_{b}$:\\
  \qquad\=
 $b := b+1$  (find the point $p_b$ with supporting line parallel to
 $p_ap_c$.)\\
\textbf{while} $\area p_ap_cp_{d+1}>
\area p_ap_cp_{d}$:\\
  \qquad\=
 $d := d+1$ (find the other point $p_d$ with supporting line parallel to
 $p_ap_c$.)\\
$
\textit{maxarea} := \max \{
\textit{maxarea},
\area p_ap_bp_cp_{d}$\}\\
\textbf{if} $\area p_ap_{a+1}p_{c+1}\le
\area p_ap_{a+1}p_{c}$: (advance $(a,c)$ to the next antipodal pair)\\
\>$a := a+1$\\
\textbf{else}:\\
\>$c := c+1$
\-\\
\textbf{until} $(a,c)=(c_0,a_0)$
\end{tabbing}
The main loop is driven by the antipodal pair $(a,c)$.
In each iteration, either $a$ or $c$ is advanced to the next vertex.
%Without the two loops for advancing $b$ and $d$, 
This is essentially
the program for reporting
all antipodal pairs of {vertices} from
\cite[Section~4.2.3]{PS}, except that we need not be careful
about getting all such pairs if
$P$ has parallel edges. 
In the two inner loops, the points $b$ and $d$ that are farthest from the line $ac$ are updated.\looseness-1

% The basic
% predicate of this algorithm is
% the comparison between signed triangle areas. The best way to
% calculate this is as a
% $2\times2$ determinant
% % wedge product
% with just two multiplications. For example,
% $\area p_cp_ap_{b+1}-
% \area p_cp_ap_{b} = \frac12(p_c-p_a)\wedge(p_{b+1}-p_b)$, where
% $\binom {x_1}{y_1} \wedge 
% \binom {x_2}{y_2} =
%   \left|
%     \begin{smallmatrix}
%       x_1&x_2\\
%       y_1&y_2\\
%     \end{smallmatrix}\right|
%   = x_1y_2-x_2y_1$.
% The area of a quadrilateral $ABCD$ is
% $\frac12|(C-A)\wedge(D-B)|$.

\section{The
Smallest Enclosing Parallelogram 
According to
Schwarz, Teich, Welzl, and Evans~\cite{stwe-fmep-94}}
\label{sec:schwarz}

The algorithm of Schwarz et al.~\cite{stwe-fmep-94} is similar in
spirit to our algorithm in constructing a sequence of parallelograms $abcd$ by advancing the
direction to which $b$ and $d$ are parallel, following
the  rotating-calipers technique.
They also sketch an application of
smallest enclosing \para s to signal compression~\cite[Section~4]{stwe-fmep-94},
and 
the appendix
%  of the technical report of Schwarz, Teich, Welzl, and
%Evans~\cite{stwe-fmep-94}
gives details about a C++
implementation. % of their algorithm.

There is one difference in the setup. We explain it with our notation:
By \Cref{2-flush} (\cite[Lemma~1]{stwe-fmep-94}), it suffices to
look for \para s where
at least one of the sides $b$ and $d$ touches a
whole edge of~$P$,
\emph{and}
at least one of the sides $a$ and $c$ touches a whole
edge of $P$.
This means that two adjacent \para\ sides must touch edges of~$P$.
Now, the algorithm of% Schwarz et al.
~\cite{stwe-fmep-94}
only considers those \anc\ parallelograms where these two sides
are %the sides
$b$ and
$c$, % touch an edge of~$P$,
like in \Cref{fig:schwarz}a.
This restriction is compensated by sweeping over an
angular range of $360^\circ$ instead of  $180^\circ$.

\Cref{fig:schwarz} illustrates a few steps of the algorithm.  After
finding the \para\ of \Cref{fig:schwarz}a and computing its area, the
algorithm of Section~\ref{sec:pseudocode}, when specialized for the
smallest containing \para, would next look at the \para\ of
\Cref{fig:schwarz}b. This \para\ is skipped in Schwarz et
al.~\cite{stwe-fmep-94} at this point, but this omission is no mistake: This
parallelogram was already considered before with rotated labels, when
$b$ touched $p_2p_3$ and $c$ touched $p_6p_7$. The next \para\ is not
shown: $d$ touches the edge $p_{12}p_{13}$ and $a$ touches
$p_{2}p_{3}$. This
\para\ is also skipped by Schwarz et al.~\cite{stwe-fmep-94} at this point,
but
it is considered later when
$b$ touches $p_{12}p_{13}$ and $c$ touches $p_{2}p_{3}$.
\Cref{fig:schwarz}b shows the next \para. It is a largest
S-\anc\ \para\ when the side $b$ is \anc, but it is not a largest
S-\anc\ \para\ when the side $a$ is \anc, because the dashed antipodal pair
$p_7p_{13}$ is not parallel to $a$ and $c$.
Hence it cannot be a largest enclosing \para.
The algorithm of
 Schwarz et al.~\cite{stwe-fmep-94} skips this \para\ and does not
 consider it at all.

\begin{figure}[htb]
  \centering
  \includegraphics{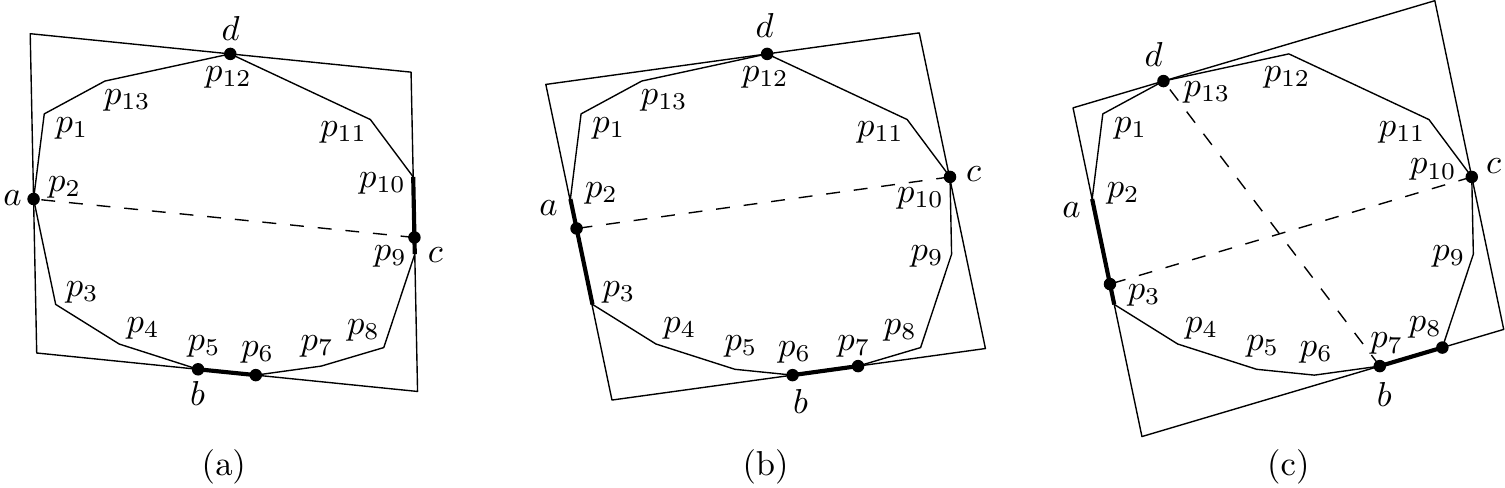}
%\vskip -2mm
  \caption{Three snapshots of the algorithm}
  \label{fig:schwarz}
\end{figure}

 This setup makes the algorithm simple and elegant:
Most case distinctions
% (i)--(iii)
of Section~\ref{sec:pseudocode} can
be eliminated, together with the flags
\textit{next\_AC} and
\textit{next\_BD}. % and one case distinction in the initialization.
Like in Section~\ref{sec:4gon},
the algorithm can be structured into two nested loops.
%Contrary to Section~\ref{sec:4gon},
The outer loop iterates over the edges of $P$
through which $b$ goes, and the inner loop updates the antipodal
pair $AC$ parallel to the direction of~$b$.

\begin{tabbing}
\quad\=\qquad\=\qquad\=\qquad\=\qquad\=\qquad\=\qquad\=\kill
\>\+
$b := 1$; $c := 2$; $d := 2$; $a := 3$\\
\textbf{while}
$\area p_c p_{c+1}p_{a+1}>
\area p_c p_{c+1}p_{a}$: (initialization)\\
  \>
 $a := a+1$ (find the opposite point $p_a$ with supporting line parallel to
 $p_c p_{c+1}$)\\
$\textit{minarea} := \infty$\\
\textbf{for} $b:= 1 \ldt n$:\\
  \>\+
%$\uu := p_{b+1}-p_{b}$\\
\textbf{while} $\area p_b p_{b+1}p_{d+1}>
\area p_b p_{b+1}p_{d}$:\\
  \>
 $d := d+1$ (update the point $p_d$ opposite to
 $p_b p_{b+1}$)\\
  \textbf{while}
$\area p_b p_{b+1}p_{a}>
\area p_b p_{b+1}p_{c+1}$:
%  $\det (\uu, p_{c+1}-p_{a})< 0$:
\+\\
$c := c+1$ \hbox spread -1ex{(search for antipodal pair $AC$ parallel to $\uu=p_bp_{b+1}$, with
$C$ on an edge)}\\
\textbf{while} \=
$\area p_c p_{c+1}p_{a+1}>
\area p_c p_{c+1}p_{a}$\\
\>\llap{\textbf{or} (}$\area p_c p_{c+1}p_{a+1}=
\area p_c p_{c+1}p_{a}$
\textbf{and}
$\area p_b p_{b+1}p_{a+1}\ge
\area p_b p_{b+1}p_{c}$):\\
\>  \quad
 $a := a+1$ (update the point $p_a$ opposite to
 $p_c p_{c+1}$)\-\\
\textbf{if} % $\det (\uu, p_c-p_{a})\ge 0$:\\
$\area p_b p_{b+1}p_{a}\ge
\area p_b p_{b+1}p_{c}$:\\
\qquad\=\+
construct the point $C$ on the line $p_cp_{c+1}$ such that $p_aC$ is
parallel to $\uu=p_bp_{b+1}$\\
$\textit{minarea} := \min \{
\textit{minarea},
2\cdot\area p_ap_bCp_{d}$\}
\end{tabbing}
The algorithm of% Schwarz et al.
~\cite{stwe-fmep-94} actually uses
a precomputed list $L=((p_i,p_{i+1}),p_{q_i})_{i=1\ldt n}$ that stores for each edge 
$(p_i,p_{i+1})$ of $P$ an antipodal vertex $p_{q_i}$ that
is farthest away from the line through
$(p_i,p_{i+1})$. By contrast, the algorithm above %does not use such a list but it
updates the  vertex $p_d$ opposite to $p_b p_{b+1}$ and the vertex $p_a$
 opposite to $p_c p_{c+1}$ on the fly.
The treatment of degenerate cases is also different.%\looseness-1

\goodbreak

% ``antipodal'' vertex edge pairs
% $(Q_i,e_i)$ of
% \Cref{lemma-AC-BD}\ref{lemma:AC} in advance.
% without merging the pairs from the two sides

%\tableofcontents

\end{document}